%% file: paper.tex
\documentclass[reprint,amsmath,amssymb,aps,pra,nofootinbib]{revtex4-2}

\usepackage[T1]{fontenc}
\usepackage[utf8]{inputenc}
\usepackage[english]{babel}
\usepackage[margin=1.8cm]{geometry}
\usepackage{graphicx}
\usepackage{ragged2e}
\usepackage{booktabs}
\usepackage{amsthm}
\usepackage{fancyhdr}
\usepackage{titlesec}
\usepackage[unicode]{hyperref}
\usepackage{orcidlink}
\usepackage{todonotes}

\renewcommand{\leq}{\leqslant}
\renewcommand{\geq}{\geqslant}
\newcommand{\deq}[1]{\mathrel{\smash[t]{\overset{#1}{=}}}}
\DeclareMathOperator{\Rad}{Rad}
\DeclareMathOperator{\GL}{GL}
\DeclareMathOperator{\RM}{RM}
\DeclareMathOperator{\GF}{GF}
\DeclareMathOperator{\wt}{wt}
\DeclareMathOperator{\St}{St}
\DeclareMathOperator{\arank}{arank}
\newcommand{\F}{\mathbb{F}}

\usepackage{titlesec}
\makeatletter
\renewcommand\thesection{\arabic{section}}
\makeatother
\titleformat{\section}[block]{\normalfont\normalsize\scshape\centering}{\thesection.}{0.75em}{}
\titlespacing*{\section}{0pt}{\baselineskip}{0.5\baselineskip}
\titlespacing*{\subsection}{0pt}{\baselineskip}{0.25\baselineskip}

\fancypagestyle{plain}{
  \fancyhf{}
  \fancyfoot[C]{\thepage}

}

\theoremstyle{plain}
\newtheorem{theorem}{Theorem}
\newtheorem{lemma}{Lemma}

\begin{document}

\setlength{\abovedisplayskip}{3pt}
\setlength{\abovedisplayshortskip}{3pt}
\setlength{\belowdisplayskip}{3pt}
\setlength{\belowdisplayshortskip}{3pt}
\setlength{\headheight}{13pt}
\setlength{\parskip}{0pt}
\emergencystretch=1em

\title{The Weight Distribution of the Third-Order Reed--Muller Code of Length 2048}

\author{
Kirill Khoruzhii\orcidlink{0000-0003-4689-3812}$^{1,*}$,
Patrick Gel\ss\orcidlink{0000-0002-3645-9513}$^{1}$,
Sebastian Pokutta\orcidlink{0000-0001-7365-3000}$^{1,2}$
}
\affiliation{
$^1$Zuse Institute Berlin, Berlin, Germany\\
$^2$Technische Universit\"at Berlin, Germany
}

\begin{abstract}
We compute the weight distribution of the third-order Reed--Muller code $\RM(3,11)$ of length $2048$. The weight enumerator is assembled from the coset weight enumerators of $f+\RM(2,10)$, evaluated for representatives of all $3\,691\,560$ nonzero $\GL(10,2)$-orbits of Boolean cubic forms in ten variables. The computation rests on a structural theorem: a nondegenerate Boolean cubic form admits a nondegenerate hyperplane restriction, except for a single orbit in each odd dimension. The same pass determines the second-order nonlinearity of every cubic form: the relative covering radius of $\RM(2,10)$ in $\RM(3,10)$ is $408$, attained on $179$ orbits. This raises the best known lower bound on the covering radius of $\RM(2,10)$ from $400$ to $408$. A complementary heuristic search shows that the relative covering radius of $\RM(6,10)$ in $\RM(7,10)$ is at most $32$, improving the previous bound of $50$.
\end{abstract}

\maketitle
\thispagestyle{fancy}

\begingroup
\renewcommand\thefootnote{\fnsymbol{footnote}}
\footnotetext[1]{khoruzhii@zib.de}
\endgroup

\section{Introduction}

The weight distributions of the Reed--Muller codes $\RM(r,m)$~\cite{muller_1954,reed_1954} are a classical subject of coding theory, yet they are known in closed form only for orders $r\leq 2$~\cite{sloane_1970} and, through the MacWilliams identities, for the dual orders $r\geq m-3$ (see~\cite{abbe_2021} for a survey). Beyond these layers, general results cover only the low-weight range~\cite{kasami_1970}, and complete distributions are known only from individual computations. For the third-order codes the known cases form a chain of length-doubling steps: $\RM(3,7)$~\cite{sugino_1971}, $\RM(3,8)$~\cite{tilborg_1971,sarwate_1973}, $\RM(3,9)$~\cite{sugita_1996}, and $\RM(3,10)$~\cite{brier_2003}; recently the fourth-order code $\RM(4,9)$ of length $512$ was enumerated~\cite{markov_2025}. This paper computes the next step of the chain, the weight distribution of the third-order Reed--Muller code $\RM(3,11)$ of length $2048$ (Tab.~\ref{tab:enum}).

The route is the Sarwate-type recursion~\cite{sarwate_1973,markov_2025}: the weight enumerator of $\RM(3,m+1)$ is the sum of squared coset weight enumerators of $f+\RM(2,m)$ over all cubic forms $f$, and since a coset enumerator depends only on the $\GL(m,2)$-orbit of $f$, the sum collapses to orbit representatives weighted by orbit sizes (Sec.~\ref{sec:rm-codes}). Every computation in the chain therefore rests on a classification one dimension lower: Hou classified the cubic forms for $m\leq 8$~\cite{hou_1996}, the case $m=9$ was classified in~\cite{brier_2003,hora_2021}, and the fourth-order computation~\cite{markov_2025} uses the classifications of quartic forms in eight variables and of the cosets of $\RM(2,7)$ in $\RM(4,7)$~\cite{gillot_2023}. For $m=10$ the required input became available only recently: the classification of Boolean cubic forms in ten variables~\cite{khoruzhii_2026b}, which provides representatives and stabilizer orders for all $3\,691\,560$ nonzero $\GL(10,2)$-orbits.

The classification alone does not make the computation feasible; the concluding remarks of~\cite{markov_2025} name the per-coset cost as the principal obstacle to any progress beyond length $512$. The split formula of Brier and Langevin~\cite{brier_2003} evaluates one coset enumerator by an outer enumeration over the homogeneous quadratic forms one dimension lower, $2^{36}$ forms for $m=10$, which is out of reach at the scale of the catalog. The second ingredient of this paper is a structural theorem: every nondegenerate Boolean cubic form admits a nondegenerate hyperplane restriction, except for one orbit in each odd dimension (Thm.~\ref{thm:nondegenerate-restriction}). Splitting along such a restriction saves a factor of $2^m$ in the enumeration, $2^{26}$ instead of $2^{36}$ cosets per orbit at $m=10$. With this speedup the full catalog pass takes about $65$ CPU-years and yields the coset weight enumerator of every orbit (Sec.~\ref{sec:coset-weight-enum}).

The lowest nonzero coefficient of the coset enumerator of $f$ is its second-order nonlinearity $d_2(f)$, a long-studied quantity in the covering-radius literature~\cite{hou_1993} and in cryptographic Boolean function analysis, and more recently a measure of Clifford approximability in fault-tolerant quantum compilation~\cite{khoruzhii_2026a}. Its worst case over cubic forms is the relative covering radius $\rho_{2,3}(m)$ of $\RM(2,m)$ in $\RM(3,m)$: the value $\rho_{2,3}(6)=18$ goes back to Schatz~\cite{schatz_1981}, the values $40$ and $88$ for $m=7,8$ follow from the classification in at most eight variables~\cite{hou_1996}, and $\rho_{2,3}(9)=196$ was computed in~\cite{brier_2003}. The enumerator pass determines $d_2$ on every orbit and hence $\rho_{2,3}(10)=408$, attained on $179$ orbits, listed in Tab.~\ref{tab:rep}. This also raises the best known lower bound on the covering radius of $\RM(2,10)$ from $400$~\cite{fourquet_2008} to $408$ (Sec.~\ref{sec:discussion}).

Two further results accompany the computation. We prove a bound (Thm.~\ref{thm:arank-d2-bound}) for cubic forms of alternating rank $r$, which forces rank at least $6$ on the extremal orbits and describes the correlation between Clifford approximability and non-Clifford resource cost~\cite{khoruzhii_2026a,khoruzhii_2026b,amy_2019}. We develop a local-search heuristic that produces explicit upper-bound certificates, one correction per orbit, reproducing $\rho_{2,3}(10)\leq 408$ at roughly a thousandth of the cost of the enumerator pass; applied to degree-seven forms, it yields $\rho_{6,7}(10)\leq 32$, improving the bound of $50$~\cite{dougherty_2022}.

\section{Reed--Muller Codes} \label{sec:rm-codes}

For $0\leq r\leq m$, the Reed--Muller code $\RM(r,m)$ consists of the Boolean polynomials of degree at most $r$ in $m$ variables, reduced modulo $x_j^2=x_j$ and read as functions $\F_2^m\to\F_2$. Listing its $2^m$ values turns a Boolean function into a vector in $\F_2^{2^m}$, so $\RM(r,m)$ is a linear code, with the lower-order codes as nested subcodes (Fig.~\ref{fig:1}a). The weight of a Boolean function counts its support,
\begin{equation}
    \wt(f)=|\{x\in\F_2^m: f(x)=1\}|,
\end{equation}
and its distance to $\RM(r,m)$,
\begin{equation*}
    d_r(f)=\min_{q\in\RM(r,m)}\wt(f+q),
\end{equation*}
is the $r$-th order nonlinearity of $f$; the case $r=2$ gives the second-order nonlinearity $d_2$. The worst case of $d_r$ over the next-order code is the relative covering radius of $\RM(r,m)$ in $\RM(r+1,m)$,
\begin{equation*}
    \rho_{r,r+1}(m)=\max_{f\in\RM(r+1,m)} d_r(f).
\end{equation*}
Maximizing over all Boolean functions instead gives the covering radius $\rho_r(m) \geq \rho_{r,r+1}(m)$ of $\RM(r,m)$.
The main covering-radius target of this paper is $\rho_{2,3}(10)$; the cocubic value $\rho_{6,7}(10)$ appears in Sec.~\ref{sec:heuristic-upper-bound}.

Adding any quadratic to $f$ leaves $d_2(f)$ unchanged, so $d_2$ descends to the quotient $\RM^*(3,m)=\RM(3,m)/\RM(2,m)$, identified with the homogeneous cubic forms. Two cubic forms are $\GL(m,2)$-equivalent when
\begin{equation*}
    f_1\sim f_2 \ \Leftrightarrow\ \exists A\in\GL(m,2): f_1(x)\deq{3} f_2(Ax),
\end{equation*}
where $g\deq{k}h$ denotes equality of the degree-$k$ homogeneous parts. Weight is $\GL(m,2)$-invariant, so $d_2$ is constant on orbits. For $m=10$ there are $3\,691\,560$ nonzero orbits, cataloged in~\cite{khoruzhii_2026b}, whose representatives we use throughout. The stabilizer of a form is
\begin{equation*}
    \St(f)=\{A\in\GL(m,2): f(Ax)\deq{3}f(x)\},
\end{equation*}
and its orbit has size $|\GL(m,2)|/|\St(f)|$. We write representatives in the monomial notation of~\cite{khoruzhii_2026b}: a string $ijk$ denotes $x_ix_jx_k$ and strings are summed over $\F_2$, so for instance $025+034$ means $x_0x_2x_5+x_0x_3x_4$.

For $u\in\F_2^m$ the finite difference $\Delta_u f(x)=f(x+u)+f(x)$ lowers the degree by one. The radical
\begin{equation*}
    \Rad(f)=\{u\in\F_2^m:\Delta_u f\deq{2}0\}
\end{equation*}
collects the directions along which $f$ reduces. A cubic form is nondegenerate when $\Rad(f)=0$. 

The weight distribution of $\RM(3,m+1)$, recorded by
\begin{equation*}
    W_{\RM(3,m+1)}(z)=\sum_{f\in\RM(3,m+1)} z^{\wt(f)},
\end{equation*}
is the main quantity of this paper. Following~\cite{sarwate_1973,brier_2003}, it is determined by the coset weight enumerators of the $\GL(m,2)$-orbits in $\RM^*(3,m)$. The coset weight enumerator of a Boolean function $f$ is
\begin{equation}
    W_f(z)=\sum_{q\in\RM(2,m)} z^{\wt(f+q)},
    \label{eq:coset-enum}
\end{equation}
and its lowest nonzero exponent is $d_2(f)$. Weights are $\GL(m,2)$-invariant and $W_f$ is unchanged by adding a quadratic to $f$, so $W_f$ depends only on the $\GL(m,2)$-orbit of $f$. The reduction then reads
\begin{equation}
    W_{\RM(3,m+1)}(z)=\sum_{f}\frac{|\GL(m,2)|}{|\St(f)|}\,W_f(z)^2,
    \label{eq:wd-reduction}
\end{equation}
summed over orbit representatives $f$, with the zero orbit contributing $W_{\RM(2,m)}(z)^2$. For $m=10$ the coset enumerators of the $3\,691\,561$ orbits thus give the complete $\RM(3,11)$ weight distribution (Tab.~\ref{tab:enum}). 

\section{Coset Weight Enumerators} \label{sec:coset-weight-enum}

In this section we evaluate the coset weight enumerator~\eqref{eq:coset-enum} for every nonzero orbit representative of the catalog~\cite{khoruzhii_2026b}. The zero orbit contributes the classical weight distribution of $\RM(2,10)$~\cite{sloane_1970}, so by~\eqref{eq:wd-reduction} this data assembles into the weight distribution of $\RM(3,11)$. The lowest nonzero exponent of each enumerator is the value of $d_2$ on its orbit, so the same pass also determines $\rho_{2,3}(10)$ and the extremal orbits.

The enumerator of a single orbit reduces to affine coset enumerators. Fix a linear coordinate $t$ and write $x=(t,y)$ with $y\in\F_2^{m-1}$. The cubic form has a unique split
\begin{equation*}
    f(t,y)=g(y)+tp(y),
\end{equation*}
where $g\in\RM^*(3,m-1)$ is the restriction of $f$ to the hyperplane $t=0$ and $p\in\RM^*(2,m-1)$. A quadratic correction $q\in\RM(2,m)$ is determined by a homogeneous quadratic form $h\in\RM^*(2,m-1)$ and two affine corrections $a_0,a_1\in\RM(1,m-1)$ on the two slices,
\begin{equation*}
    f(t,y)+q(t,y)=\left\{\begin{aligned}
        &g(y)+h(y)+a_0(y), & t&=0,\\
        &g(y)+p(y)+h(y)+a_1(y), & t&=1,
    \end{aligned}\right.
\end{equation*}
and conversely every triple $(h,a_0,a_1)$ arises from a unique $q$. Thus, if
\begin{equation*}
    A_u(z)=\sum_{a\in\RM(1,m-1)} z^{\wt(u+a)}
\end{equation*}
is the affine coset enumerator of a Boolean function $u$ on $\F_2^{m-1}$, then the split gives~\cite{brier_2003}
\begin{equation}
    W_f(z)=\sum_{h\in\RM^*(2,m-1)} A_{g+h}(z)\,A_{g+p+h}(z).
    \label{eq:split-weight-enumerator}
\end{equation}
The weights $\wt(u+a)$ over an affine coset are read off the Walsh spectrum of $u$, so each factor is computed by a fast Walsh--Hadamard transform on $2^{m-1}$ points. For $m=10$, the split reduces the outer enumeration from the $|\RM(2,10)|=2^{56}$ corrections $q$ to the $|\RM^*(2,9)|=2^{36}$ forms $h$. Whenever the lowest nonzero bin of $W_f$ is attained, we keep a triple $(h,a_0,a_1)$ attaining it, which reconstructs an explicit quadratic correction $q$ with $\wt(f+q)=d_2(f)$.

The product in~\eqref{eq:split-weight-enumerator} is unchanged under two translations of $h$. First, replacing $h$ by $h+p$ swaps the two factors. Second, write $\Delta_i u=u(y+e_i)+u(y)$ for a coordinate direction $e_i$ of the $y$-space; since $A_u$ is invariant under translations of $y$ and under adding affine functions, replacing $h$ by $h+\Delta_i g$ leaves both factors unchanged. Therefore the product is constant on the cosets of
\begin{equation*}
    V_f=\langle p,\Delta_0 g,\ldots,\Delta_{m-2}g\rangle\subseteq\RM^*(2,m-1),
\end{equation*}
and~\eqref{eq:split-weight-enumerator} factors as
\begin{equation}
    W_f(z)=|V_f|\sum_{h\in\RM^*(2,m-1)/V_f} A_{g+h}(z)\,A_{g+p+h}(z),
    \label{eq:quotient-enumerator}
\end{equation}
with one representative $h$ chosen from each coset of $V_f$.

The size of the quotient is controlled by the radical of the slice: for a nondegenerate form $f$, every split satisfies $\dim V_f=m-\dim\Rad(g)$ (App.~\ref{app:nondegenerate-restrictions}), so a split with nondegenerate $g$ has $\dim V_f=m$ and saves a factor $2^m$ in the enumeration. The following theorem guarantees that such a split exists.

\begin{theorem} \label{thm:nondegenerate-restriction}
Every nondegenerate Boolean cubic form on $\F_2^m$ has a nondegenerate hyperplane restriction, except for one orbit in each odd dimension: the forms equivalent to $f_*(x)=x_0(x_1x_2+x_3x_4+\ldots+x_{m-2}x_{m-1})$ have no nondegenerate hyperplane restriction.
\end{theorem}

The proof is given in the Appendix. The dimension $m=10$ is even, so there is no exception, and the evaluation of~\eqref{eq:quotient-enumerator} runs over $2^{26}$ coset representatives instead of $2^{36}$ forms for every nondegenerate orbit; without this reduction the full catalog pass would be out of reach. Degenerate forms need no separate treatment: the $348$ degenerate nonzero orbits, lifted from the orbits in dimension at most $9$, are processed by the same computation with a smaller $V_f$ and a correspondingly larger quotient.

We evaluated~\eqref{eq:quotient-enumerator} for all $3\,691\,560$ nonzero orbit representatives. All timings in this paper are summed single-thread CPU times measured on $48$-core Intel Xeon Gold 6246 nodes. The pass is trivially parallel across orbits and took $64.7$ CPU-years, averaging $9.2$ CPU-minutes per orbit; the degenerate orbits account for only $126$ CPU-hours of this total. Two checks support the result: the same computation one dimension lower reproduces the known weight distribution of $\RM(3,10)$~\cite{brier_2003}, and the coefficients of the assembled enumerator sum to $2^{232}=|\RM(3,11)|$. Combining the enumerators via~\eqref{eq:wd-reduction}, with the orbit sizes $|\GL(10,2)|/|\St(f)|$ taken from the catalog~\cite{khoruzhii_2026b} and the zero orbit contributing $W_{\RM(2,10)}(z)^2$, gives the main result.

\begin{theorem}
The weight distribution of the third-order Reed--Muller code of length $2048$ is as given in Tab.~\ref{tab:enum}.
\end{theorem}

The lowest nonzero bin of each coset enumerator is the value of $d_2(f)$ on its orbit. The maximum over the catalog is $408$, attained on $179$ orbits; their representatives are listed in Tab.~\ref{tab:rep} together with the kissing number
\begin{equation*}
    K=|\{q\in\RM(2,10):\wt(f+q)=408\}|.
\end{equation*}
All $179$ forms are nondegenerate, so $|V_f|=2^{10}$ and every $K$ is divisible by $2^{10}$; the table lists $K/2^{10}$.

\begin{theorem}
The relative covering radius of $\RM(2,10)$ in $\RM(3,10)$ is $408$.
\end{theorem}

\section{Connection with Alternating Rank} \label{sec:arank-d2}

We now relate $d_2$ to the alternating rank (arank) of a cubic form, defined~\cite{khoruzhii_2026b} as the smallest $r$ such that
\begin{equation*}
    f(x)\deq{3}\sum_{t=1}^r u_t(x)v_t(x)w_t(x)
\end{equation*}
for linear forms $u_t,v_t,w_t$. Both invariants have natural readings in fault-tolerant quantum compilation: $\arank(f)$ is the non-Clifford cost of realizing the phase polynomial $(-1)^{f(x)}$~\cite{khoruzhii_2026a}, and $d_2(f)$ measures its approximability by a Clifford. They quantify distance from the Clifford layer in complementary senses. The bound below makes one direction precise: adding a rank-one cubic to $f$ can only increase $d_2(f)$ by a controlled amount.

\begin{lemma}
\label{lem:rank-one-step}
Let $f$ be a Boolean function on $\F_2^m$. Then $d_2(f+\varphi) \leq 2^{m-3}+\frac34 d_2(f)$ for every cubic form $\varphi$ of alternating rank one.
\end{lemma}

\begin{proof}
Write $\varphi=\ell_1\ell_2\ell_3$ with independent linear forms. Let $q\in\RM(2,m)$ attain $d_2(f)$, so $E=\{x:f(x)+q(x)=1\}$ has size $|E|=d_2(f)$. For each $\alpha=(\alpha_1,\alpha_2,\alpha_3)\in\F_2^3$, define
\begin{equation*}
    \varphi_\alpha=(\ell_1+\alpha_1)(\ell_2+\alpha_2)(\ell_3+\alpha_3).
\end{equation*}
The cubic part of $\varphi_\alpha$ equals $\varphi$, so $\varphi+\varphi_\alpha\in\RM(2,m)$. Its support is the cell $C_\alpha=\{x:\ell_i(x)=1+\alpha_i,\ i=1,2,3\}$, one of eight cells of size $2^{m-3}$ partitioning $\F_2^m$. Adding $\varphi_\alpha$ to $f+q$ flips values on $C_\alpha$, so
\begin{equation*}
    d_2(f+\varphi)\leq\wt(f+q+\varphi_\alpha)=|E|+2^{m-3}-2|C_\alpha\cap E|.
\end{equation*}
The eight intersections $C_\alpha\cap E$ partition $E$, so some has size at least $|E|/8$, giving
\begin{equation*}
    d_2(f+\varphi)\leq 2^{m-3}+\tfrac{3}{4}|E|,
\end{equation*}
as claimed.
\end{proof}

A cubic form of alternating rank $r$ is a sum of $r$ arank-1 cubics, so iterating Lemma~\ref{lem:rank-one-step} gives a closed-form bound.

\begin{theorem}
\label{thm:arank-d2-bound}
Let $f$ be a Boolean cubic form on $\F_2^m$ of alternating rank $r$. Then its second-order nonlinearity satisfies $d_2(f) \leq 2^{m-1}\left(1-\left(\frac34\right)^r\right)$.
\end{theorem}

\begin{proof}[Proof of Theorem~\ref{thm:arank-d2-bound}]
Put $f_k=\varphi_1+\cdots+\varphi_k$ and $D_k=d_2(f_k)$, where each $\varphi_i$ has alternating rank one. Since $D_0=0$, Lemma~\ref{lem:rank-one-step} gives $D_k\leq 2^{m-3}+\frac{3}{4}D_{k-1}$. Iterating this recurrence yields
\begin{equation*}
    D_r\leq 2^{m-3}\sum_{j=0}^{r-1}\left(\frac{3}{4}\right)^j
        =2^{m-1}\left(1-\left(\frac{3}{4}\right)^r\right).
\end{equation*}
Since $D_r=d_2(f)$, the result follows.
\end{proof}

For any cubic form $f$ with $d_2(f)=408$, Theorem~\ref{thm:arank-d2-bound} forces $\arank(f)\geq 6$. Together with the maximum $\arank=7$ in $m=10$~\cite{khoruzhii_2026b}, this gives $\arank(f)\in\{6,7\}$. The dimension $m=10$ is the first for which the bound admits more than one value: in dimensions $m\leq 9$, substituting $\rho_{2,3}(m)$~\cite{schatz_1981,hou_1996,brier_2003} forces $\arank(f)$ equal to the maximum value attained in that dimension~\cite{khoruzhii_2026b}. Both possibilities occur at $m=10$: of the $179$ orbits with $d_2=408$ found in Sec.~\ref{sec:coset-weight-enum}, $177$ have $\arank=7$ and $2$ have $\arank=6$ (Tab.~\ref{tab:rep}).

\begin{figure*}[t]
    \centering
    \includegraphics{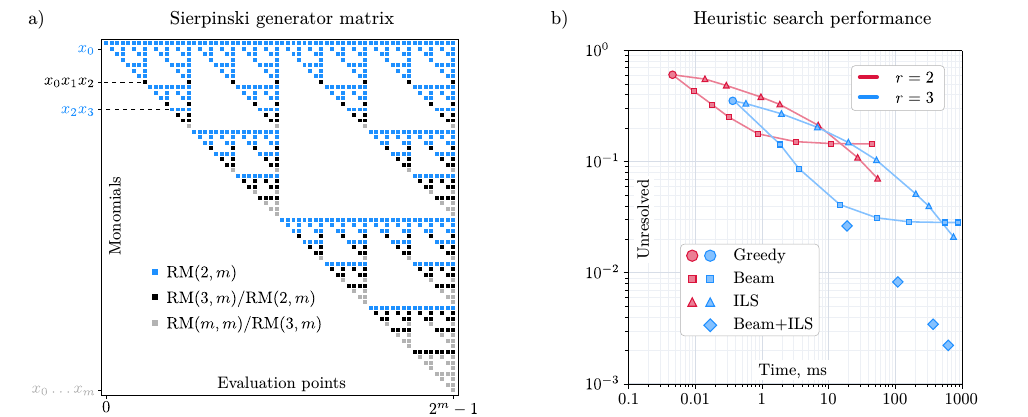}
    \caption{\justifying
        \textbf{Reed--Muller structure and heuristic search.}
        a) Generator matrix of $\RM(3,m)$, with the subcodes highlighted.
        b) Calibration of the upper-bound heuristic on a hard sample of $10^4$ catalog representatives. An instance is unresolved if the best correction found has weight greater than $400$.
    }
    \label{fig:1}
\end{figure*}

\section{Heuristic Upper Bounds} \label{sec:heuristic-upper-bound}

The computation of Sec.~\ref{sec:coset-weight-enum} settles $m=10$, but its cost is dominated by the $2^{26}$-fold enumeration repeated on every orbit. An upper bound on the covering radius needs much less: a single correction $q$ with small $\wt(f+q)$ per orbit is a certificate, found by local search and verified by one weight evaluation. This section describes such a search. It reproduces the bound $\rho_{2,3}(10)\leq 408$ at roughly a thousandth of the cost of the enumeration pass, and it extends to settings where no enumeration is available: applied to degree-seven forms, it yields $\rho_{6,7}(10)\leq 32$, improving the previous bound of $50$~\cite{dougherty_2022}.

For a fixed representative $f$, the coefficients of a correction $q\in\RM(2,10)$ form a $56$-dimensional vector: one constant, $10$ linear, and $45$ quadratic coefficients. The local search varies only the nonconstant coordinates; whenever a candidate is evaluated, the constant term is set to the value giving the smaller $\wt(f+q)$. The effective search space is thus $\F_2^{55}$, indexed by the linear and quadratic monomials.

We initialize each search at $q=0$. The baseline is one-flip greedy descent: at each step we try all single coefficient flips and accept the one giving the largest decrease of $\wt(f+q)$, stopping when no improving flip remains. Three extensions improve this descent. First, a $\tau$-move flips any nonempty set of at most $\tau$ coefficients. Second, after reaching a local endpoint, a random kick ($k$ random $\tau$-moves) followed by a new descent gives, after $i$ iterations, an iterated local search (ILS). Third, instead of keeping only the best improving move, one can keep the best $w$ improving descendants and continue from this beam. The best-performing routine combines all three ideas and is specified by the parameter tuple $(\tau,w,i,k)$.

For calibration we used a hard sample of $10^4$ catalog representatives. Preliminary runs showed that reaching weight $408$ was relatively easy, whereas the threshold $400$ produced a useful hard tail. We therefore built the sample from representatives for which a stronger but slower preliminary search had already found a correction of weight $400$, and measured how quickly the tested routines could rediscover such a correction; an instance was called unresolved if the best correction found still had weight greater than $400$. Fig.~\ref{fig:1}b plots the unresolved fraction against the running time for greedy descent, beam descent, ILS, and the combined $(\tau,w,i,k)$ search; for the combined search we swept the relevant $(w,i)$ choices and kept the best points at each time scale. The comparison showed that $\tau=2$ is preferable at very small budgets, while $\tau=3$ gives a better hard tail once more time is available, so the combined search with $\tau=3$ was used for the full catalog pass.

The full pass with $(\tau,w,i,k)=(3,64,24,1)$ found a correction of weight at most $408$ for every one of the $3\,691\,560$ nonzero representatives in $538$ CPU-hours, about three orders of magnitude below the enumeration pass of Sec.~\ref{sec:coset-weight-enum}, and left $3038$ representatives on the boundary $\wt(f+q)=408$. A second pass on the boundary with a longer ILS part $i=384$ and larger kicks $k=4$ reduced the boundary to $258$ representatives in another $14$ CPU-hours. The enumerators measure how tight this cheap search is: of the $258$ orbits where it stalled at $408$, the $179$ orbits of Tab.~\ref{tab:rep} indeed have $d_2(f)=408$, while the remaining $79$ have $d_2(f)<408$.

The same framework gives the cocubic bound. Here the corrections range over $\RM(6,10)$. Complementing monomials~\cite{dougherty_2022} identifies the $\GL(10,2)$-orbits in $\RM^*(7,10)$ with those in $\RM^*(3,10)$, so the same catalog provides all $3\,691\,560$ degree-seven representatives. The combined search with corrections in $\RM(6,10)$ found, for every representative $g\in\RM^*(7,10)$, a polynomial $p\in\RM(6,10)$ with $\wt(g+p)\leq 32$.

\begin{lemma}
The relative covering radius of $\RM(6,10)$ in $\RM(7,10)$ is at most $32$.
\label{lem:rm6}
\end{lemma}

The cocubic pass took $1705$ CPU-hours. The move set, adapted to the degree-seven geometry, is described in App.~\ref{app:cocubic-search}.

\input{tab1.tex}

\input{tab2.tex}

\section{Discussion} \label{sec:discussion}

Theorem~\ref{thm:nondegenerate-restriction} is not specific to $m=10$: it guarantees the full $2^m$ quotient in every dimension, up to the single exceptional orbit in odd ones. At $m=11$ this cuts the coset enumeration from $2^{45}$ homogeneous quadratic forms to $2^{34}$ cosets, which is feasible per representative. For comparison, the current record $\rho_2(11)\geq 856$ was established by computing $d_2=856$ for the cubic form $\mathrm{Tr}(x^7)$ on $\GF(2^{11})$, reported as taking about six days~\cite{gao_2026}; the quotient enumeration of Sec.~\ref{sec:coset-weight-enum} produces the full coset weight enumerator of the same form in $60$ CPU-hours. 

The relative results bear on the covering radius of $\RM(2,m)$ itself. For $m\leq 7$ the covering radius $\rho_2(m)$ is attained on cubic forms: $\rho_2(6)=18$~\cite{schatz_1981} and $\rho_2(7)=40$~\cite{wang_2019}. For $m=9,10,11$ the best known lower bounds likewise come from cubic forms~\cite{brier_2003,fourquet_2008,gao_2026}; in particular, the previous bound $\rho_2(10)\geq 400$ is attained by cubic trace monomials~\cite{fourquet_2008}. The present computation raises it to $\rho_2(10)\geq 408$, with the $179$ orbits of Tab.~\ref{tab:rep} as explicit witnesses and as the natural candidates for the extremal functions.

Two structural signatures of the extremal orbits suggest where to look for high-$d_2$ forms in $m>10$. The first signature is the alternating rank. Thm.~\ref{thm:arank-d2-bound} forces $\arank(f)\geq 6$ on every $m=10$ extremal orbit, and $177$ of the $179$ sit at the maximal $\arank=7$. Thus, high-$d_2$ candidates appear to be strongly concentrated among forms of maximal alternating rank.

The second signature is the trace structure. Every cubic trace monomial in ten variables has $d_2\leq 400$ (App.~\ref{app:trace}), so the monomial family behind the best known lower bounds stops short of the extremal set, and only one of the $179$ extremal orbits is a two-term trace polynomial $\mathrm{Tr}(\omega x^{21}+x^{49})$. This shift from monomials to binomials suggests searching for high-$d_2$ forms in $m>10$ among sparse trace polynomials with two or more terms.

The enumeration data and reference code are archived on Zenodo, \href{https://zenodo.org/records/20773273}{\texttt{zenodo.org/records/20773273}}~\cite{bcf10_zenodo}. The archive contains the coset weight enumerators for $f+\RM(2,10)$, the assembled weight enumerators of $\RM(3,10)$ and $\RM(3,11)$, the explicit distance-$408$ quadratic corrections, and the cocubic certificates. It also includes reference implementations, in particular the exact $\RM(2,10)$ coset-enumerator code. The archive is shared with the classification paper~\cite{khoruzhii_2026b}; therefore it also contains the catalog data used here, such as packed representatives and stabilizer orders, together with verification scripts for the classification data. The accompanying repository, \href{https://github.com/khoruzhii/bcf10}{\texttt{github.com/khoruzhii/bcf10}}, contains supplementary code from~\cite{khoruzhii_2026b}, as well as the heuristic upper-bound search, the $m=11$ coset-enumeration example, and the trace-polynomial checks.

\section*{Acknowledgments}

This research was supported by the DFG Cluster of Excellence MATH+ (EXC-2046/2, project id 390685689) funded by the Deutsche Forschungsgemeinschaft (DFG), as well as by the National High-Performance Computing (NHR) network.

\section{Appendix} \label{app:appendix}

\subsection{Proof of Nondegenerate Restrictions} \label{app:nondegenerate-restrictions}

\begin{lemma}
Fix a coordinate split $x=(t,y)\in\F_2\times\F_2^{m-1}$ and write $f(t,y)=g(y)+tp(y)$ for a nondegenerate cubic form $f$ on $\F_2^m$. Put $V_f=\langle p,\Delta_a g\rangle$  with $a\in\F_2^{m-1}$. Then $\dim V_f=m-\dim\Rad(g)$. In particular, if $g$ is nondegenerate, then $\dim V_f=m$.
\end{lemma}

\begin{proof}
The map $a\mapsto\Delta_a g$ from the $y$-space to $\RM^*(2,m-1)$ has kernel $\Rad(g)$, so its image has dimension $m-1-\dim\Rad(g)$. It remains to show that $p$ is not in this image. Suppose, to the contrary, that $p\deq{2}\Delta_a g$ for some $a$. Then
\begin{equation*}
    \Delta_{(1,a)}f\deq{2}\Delta_a g+p+t\Delta_a p.
\end{equation*}
The first two terms cancel, and $\Delta_a p\deq{1}\Delta_a\Delta_a g=0$. Hence $\Delta_{(1,a)}f\deq{2}0$, so $(1,a)\in\Rad(f)$, contradicting the nondegeneracy of $f$. Thus $p$ contributes one additional independent generator.
\end{proof}

The useful case for the enumerator is therefore a split with nondegenerate $g$. The remaining question is whether such a split can always be found. Equivalently, can a nondegenerate cubic form have only degenerate cubic restrictions in every coordinate split? The next lemma shows that this failure is rigid: it happens only in odd dimension and only for one $\GL(m,2)$-orbit.

\begin{lemma}
Let  $f$ be a nondegenerate cubic form on $\F_2^m$ with $m>3$. If every hyperplane restriction of $f$ is degenerate, then $m$ is odd and there is an $A\in\GL(m,2)$ such that $f(Ax)\deq{3} f_* (x) = x_0(x_1x_2+x_3x_4+\cdots+x_{m-2}x_{m-1})$.
\end{lemma}

\begin{proof}
Fix a nonzero linear form $\lambda$, determining the hyperplane $\ker\lambda = \{x : \lambda(x) = 0 \}$. Since the restriction is degenerate, there is a nonzero $a\in\ker\lambda$ such that
\begin{equation*}
     \Delta_a f |_{\ker \lambda} \deq{2} 0.
\end{equation*}
and therefore 
\begin{equation*}
    \Delta_a f \deq{2} \lambda \mu 
\end{equation*}
for some linear form $\mu\notin\langle\lambda\rangle$. For each such factorization, keep the two-dimensional space $P_a=\langle\lambda,\mu\rangle$ of linear factors. These spaces cover all nonzero linear forms, because the construction starts from an arbitrary $\lambda$.

We claim that these spaces are pairwise intersecting. Suppose, to the contrary, that two of them satisfy $P_a\cap P_b=0$. The linear part of $\Delta_b\Delta_a f=\Delta_a\Delta_b f$ lies in both $P_a$ and $P_b$, hence $\Delta_b\Delta_a f\deq{1}0$. Writing $\Delta_a f\deq{2}\alpha\beta$ with $P_a=\langle\alpha,\beta\rangle$, gives 
\begin{equation*}
    \alpha(b)\beta+\beta(b)\alpha=0,
\end{equation*}
so $\nu(b)=0$ for all $\nu\in P_a$. Similarly, $\nu(a)=0$ for all $\nu\in P_b$. Also $\Delta_a\Delta_a f=\Delta_b\Delta_b f=0$, so $\nu(a)=0$ for all $\nu\in P_a$ and $\nu(b)=0$ for all $\nu\in P_b$. Thus
\begin{equation*}
    \nu(a)=\nu(b)=0
\end{equation*}
for every $\nu\in P_a+P_b$.
Since $a$ and $b$ are distinct nonzero vectors, choose a linear form $\eta$ with $\eta(a)=\eta(b)=1$. By the covering property, $\eta$ lies in some space $P$. The space $P$ cannot be disjoint from either $P_a$ or $P_b$, because otherwise every form in $P$ would be zero on $a$ or on $b$, contradicting $\eta(a)=\eta(b)=1$. Thus $P\cap P_a\ne0$ and $P\cap P_b\ne0$. Since $P_a\cap P_b=0$ and $\dim P=2$, these two intersections force $P\subseteq P_a+P_b$. But then $\eta\in P_a+P_b$, contradicting $\eta(a)=\eta(b)=1$. Therefore the spaces are pairwise intersecting.

A pairwise-intersecting collection of two-dimensional spaces either has a common nonzero form or is contained in a three-dimensional space. Since our spaces cover all nonzero linear forms and $m>3$, the second alternative is impossible. Hence all spaces contain a common nonzero linear form $\ell$.

Since the supports cover all nonzero linear forms and all contain $\ell$, for every $\lambda\notin\langle\ell\rangle$ there is a witness $a \in \ker \ell$ such that
\begin{equation*}
    \Delta_{a}f\deq{2}\ell\lambda.
\end{equation*}
Choose linear forms $\lambda_1,\ldots,\lambda_{m-1}$, forming with $\ell$ a basis, with corresponding witnesses $a_1,\ldots,a_{m-1}\in\ker\ell$. The quadratics $\ell\lambda_i$ are independent, so the vectors $a_i$ are independent because $f$ is nondegenerate. Hence $a_1,\ldots,a_{m-1}$ form a basis of $\ker\ell$. Since the quadratic part of $\Delta_a f$ depends linearly on $a$, every quadratic derivative in a direction $a\in\ker\ell$ is a linear combination of the quadratics $\ell\lambda_i$. Hence each such derivative has a factor $\ell$, and $\Delta_a f|_{\ker\ell}\deq{2}0$ for all $a\in\ker\ell$. Thus $f|_{\ker \ell} \deq{3} 0$.

Choose coordinates with $x_0=\ell$. Therefore
\begin{equation*}
    f\deq{3}x_0\omega
\end{equation*}
for some homogeneous quadratic form $\omega$ in the remaining $m-1$ variables. The nondegeneracy of $f$ forces $\omega$ to be nondegenerate. Thus $\omega$ is a nondegenerate alternating quadratic form on an $(m-1)$-dimensional space. Hence $m-1$ is even, and by the symplectic normal form over $\F_2$ there is a further linear change of coordinates such that
\begin{equation*}
    \omega\deq{2}x_1x_2+x_3x_4+\cdots+x_{m-2}x_{m-1}.
\end{equation*}
This gives the required $A\in\GL(m,2)$.
\end{proof}

The cases $m<3$ are vacuous. For $m=3$, the space of cubic forms is one-dimensional, and its nonzero element is equivalent to $x_0x_1x_2$. This form is nondegenerate, and every hyperplane restriction has zero cubic part, so it is the exceptional orbit. Assume $m>3$. If a nondegenerate cubic form has no nondegenerate hyperplane restriction, the previous lemma shows that $m$ is odd and the form is equivalent to $f_*$. Thus every other nondegenerate cubic form has a nondegenerate hyperplane restriction.

It remains to check that $f_*$ is exceptional. First, $f_*$ is nondegenerate. For $(\alpha,v)\in\F_2\times\F_2^{m-1}$,
\begin{equation*}
    \Delta_{(\alpha,v)} f_*\deq{2}\alpha\omega+x_0\Delta_v\omega.
\end{equation*}
The two summands have different $x_0$-degree, so the quadratic part can vanish only if $\alpha=0$ and $\Delta_v\omega\deq{1}0$. Since $\omega$ is nondegenerate, this forces $v=0$. Therefore $\Rad f_* = 0$.

Now fix any coordinate split $x=(t,y)\in\F_2\times\F_2^{m-1}$ and write $f_*(t,y)=g(y)+tp(y)$. If $t=x_0$, then $g=0$. Otherwise choose coordinates on $t=0$ so that one coordinate is the restriction of $x_0$ and the remaining coordinate space is $L=\{t=x_0=0\}$. In these coordinates,
\begin{equation*}
    g\deq{3}x_0\,\omega|_L.
\end{equation*}
Since $\dim L=m-2$ is odd, $\omega|_L$ has a nonzero radical, so $g$ is degenerate.

\subsection{Cocubic Upper-Bound Search} \label{app:cocubic-search}

Here we describe the auxiliary computation behind the bound for $\RM(6,10)$ in $\RM(7,10)$. For each cubic representative $f_j\in\RM^*(3,10)$, we form the complementary representative $g_j\in\RM^*(7,10)$ by replacing each monomial $x_i x_j x_k$ with the product of the other seven variables. These are representatives of all nonzero $\GL(10,2)$-orbits in $\RM^*(7,10)$~\cite{dougherty_2022}.

For the cocubic search we used a different move space. Instead of flipping monomial coefficients in $\RM(6,10)$, a move adds a polynomial of the form
\begin{equation*}
    h(x)=\prod_{s=1}^{6}\ell_s(x),
\end{equation*}
where the $\ell_s$ are affine linear forms with independent linear parts. Preliminary tests showed that these geometric moves were more effective than moves in the monomial basis. The full move set has about $3.4 \times 10^9$ distinct moves, so we do not scan the whole neighborhood. Instead we sample $2048$ moves. A sampled move is produced from the current polynomial $g+p$ by choosing five random affinely independent points from the set $\{x:g(x)\neq p(x)\}$. These points span a $4$-dimensional affine subspace of $\F_2^{10}$. We take $h$ to be the polynomial of the form above whose support $\{x:h(x)=1\}$ is this affine span. This biases the search toward large support overlap. The local search then uses the same beam and ILS approach as in Sec.~\ref{sec:heuristic-upper-bound}, but with sampled moves.

The full pass used $(w,i,k)=(16,\ 2048,\ 1)$ and $2048$ sampled moves per beam node. It completed for all $3691560$ nonzero representatives. No representative remained above weight $32$. For each representative $g_j$, the computation stores an explicit correction $p_j\in\RM(6,10)$.

\subsection{Trace Representations of the Extremal Forms} \label{app:trace}

Here we test whether the $179$ extremal orbits admit compact finite-field descriptions. Identify $\F_2^{10}$ with $\GF(2^{10})=\F_2[\alpha]/(\alpha^{10}+\alpha^3+1)$ and write
\begin{equation*}
    \mathrm{Tr}(a)=a+a^2+\cdots+a^{2^9}
\end{equation*}
for the absolute trace. For an exponent $d$ of binary weight three, the function $x\mapsto\mathrm{Tr}(\lambda x^d)$ is a Boolean cubic form and Frobenius conjugation $d\mapsto 2d$ partitions these exponents into $12$ cyclotomic classes.
The previous lower bound $\rho_2(10)\geq 400$ is attained by such trace monomials~\cite{fourquet_2008}, so it is natural to ask whether the extremal orbits contain them.

A candidate is tested by evaluating the complete $\GL(10,2)$-invariant of~\cite{khoruzhii_2026b}. The invariant takes distinct values on all orbits, so a matching value identifies the orbit. Scanning all $12\cdot (2^{10}-1)$ trace monomials $\mathrm{Tr}(\lambda x^d)$ with $\lambda\neq 0$ shows that each exponent class produces one or two $\GL(10,2)$-orbits, all with $d_2\in\{352,360,400\}$. No trace monomial reaches $408$.

We then scanned all two-term trace polynomials $\mathrm{Tr}(\lambda x^d+\mu x^e)$ over the same exponent classes and all nonzero coefficient pairs, $66\cdot (2^{10}-1)^2=69\,070\,914$ candidates in total. It reaches a single $d_2=408$ orbit,
\begin{equation*}
    f=\mathrm{Tr}(\omega x^{21}+x^{49}),
\end{equation*}
where $\omega=\alpha^2+\alpha^3+\alpha^5+\alpha^6+\alpha^7$ satisfies
\begin{equation*}
    \omega^2+\omega+1=0,
\end{equation*}
a primitive cube root of unity. The catalog representative of this orbit is $025+028+047+128+129+136+148+156+179+236+359+378+469+568+678$, with $|\St(f)|=5$ and $\arank(f)=7$. The remaining $178$ extremal orbits admit no trace representation with at most two terms over the cubic exponent classes.

\bibliography{references.bib}

\end{document}

%% file: tab1.tex
\begin{table*}
\caption{Nonzero weight multiplicities of $\RM(3,11)$ up to complement symmetry. The table gives the number of codewords of weight $w$ for $0\leq w\leq 1024$; omitted weights in this range have multiplicity zero, and the entry for weight $2048-w$ is the same.}
\label{tab:enum}
\centering
\renewcommand{\arraystretch}{0.95}
\begin{tabular}{rr}
\toprule
$w$ & number of codewords of weight $w$ \\
\midrule
$0$ & $1$ \\
$256$ & $407647768$ \\
$384$ & $98572491484544$ \\
$448$ & $276678901812617216$ \\
$480$ & $20615149353221226496$ \\
$496$ & $61905903043104210944$ \\
$512$ & $376732003274980265308$ \\
$528$ & $61905903043104210944$ \\
$544$ & $36067676304128622985216$ \\
$576$ & $4744067074463987615916032$ \\
$592$ & $58004870978198589344317440$ \\
$608$ & $1337472133556209119154667520$ \\
$624$ & $9396613326131873902296563712$ \\
$640$ & $224515947039561245253371260800$ \\
$656$ & $5830645939509459951889648975872$ \\
$664$ & $2419077324635505387171301294080$ \\
$672$ & $143171775339255975773371298217984$ \\
$680$ & $118534788907139763971393763409920$ \\
$688$ & $4172133049353846806513019732885504$ \\
$696$ & $6530702417407652710233456393584640$ \\
$704$ & $142049722384193638439855542034857984$ \\
$712$ & $370985666710560047000945480957952000$ \\
$720$ & $6572128354504685192324577309902241792$ \\
$728$ & $28716300050295011531014790520420433920$ \\
$736$ & $415629184674851549258147208655812952064$ \\
$744$ & $2690951593078403784716729428137071345664$ \\
$752$ & $37827275484769393002546972405932606095360$ \\
$760$ & $370900717708670977732674519658689233682432$ \\
$768$ & $5505529296214564743102269157190400156746664$ \\
$776$ & $79894247600845975291315175037253158124388352$ \\
$784$ & $1450305177493270675359691374214264637850910720$ \\
$792$ & $29235356704790566516562551934810916714341990400$ \\
$800$ & $665799808379784560217854189229920478427489501184$ \\
$808$ & $16063992125312051655233396041061195367979885264896$ \\
$816$ & $389343899216758993477153060624106334133347655090176$ \\
$824$ & $8999235314122324085793822909334985019312993120288768$ \\
$832$ & $191323551649196821711001947792058498033273954293481472$ \\
$840$ & $3658740594851885600503798449378675081425867112456388608$ \\
$848$ & $62180730085983850326083175198032893774117004267936022528$ \\
$856$ & $933627921643138541428942307297065262293476936162893365248$ \\
$864$ & $12352904069973355408081004302523980802074961354772346568704$ \\
$872$ & $143893873726630427501277688558916646342435284778499147038720$ \\
$880$ & $1475434974947081710061936957117530055868461154884705079787520$ \\
$888$ & $13318672685563145171399602295039475393147723949626043926577152$ \\
$896$ & $105868711887026301874770702718612016167682651408928401326236416$ \\
$904$ & $741226717029475749508906553789588949003364241669009433251807232$ \\
$912$ & $4572123337168224454714897939810929009585250051627595597537607680$ \\
$920$ & $24852350650395129270609343467555813365645108334460498873040240640$ \\
$928$ & $119066951769505173151986670780366771667129702286785938434899836928$ \\
$936$ & $502890594559615896984236136146693696560200753886279445577002385408$ \\
$944$ & $1872800847399495130252973899847591197256947718238294722030933639168$ \\
$952$ & $6150577569231431151720935138862873464324233440424302381575971012608$ \\
$960$ & $17815991389962934481252993676594060919773199784241432513439128895488$ \\
$968$ & $45522975052166290589248064985562854640499044441971806982761147793408$ \\
$976$ & $102619006074996608740719830566521478416265414163140030900851676545024$ \\
$984$ & $204100811417817928738018318757436575622758854101050048337198119387136$ \\
$992$ & $358193705074942017878606080764562076052212316814353868657935396110336$ \\
$1000$ & $554724454472328610642650250933921456820000687270430933160620818694144$ \\
$1008$ & $758134074338700783303792829592026102708636982915591852455410408620032$ \\
$1016$ & $914409270550460874661954063324503188329102135989979060750359088594944$ \\
$1024$ & \phantom{42}$973354525080350422090405188363407729099289656148339075855691911557574$ \\
\bottomrule
\end{tabular}
\end{table*}

%% file: tab2.tex
\begin{table*}
\caption{Representatives of all 179 $\GL(10,2)$-orbits in $\RM^*(3,10)$ with
$d_2(f)=408$. Here $K=|\{q\in\RM(2,10): \wt(f+q)=408\}|$. ${}^{\dagger}$Forms of alternating rank $6$ (two in total); unmarked forms have alternating rank $7$.}
\label{tab:rep}
\centering
\scriptsize
\renewcommand{\arraystretch}{0.87}
\newcommand{\repsep}{\texorpdfstring{\raisebox{1.5pt}{\textbf{\scalebox{0.4}{+}}}}{+}}
\begin{minipage}[t]{0.495\textwidth}
\begin{tabular}{lr}
\toprule
representative & $K/2^{10}$ \\
\midrule
$037\repsep 046\repsep 056\repsep 058\repsep 129\repsep 136\repsep 148\repsep 157\repsep 238\repsep 247\repsep 248\repsep 256$ & 28416 \\
$023\repsep 028\repsep 037\repsep 139\repsep 146\repsep 158\repsep 234\repsep 267\repsep 356\repsep 369\repsep 478\repsep 689\smash{^\dagger}$ & 19712 \\
$023\repsep 045\repsep 146\repsep 158\repsep 179\repsep 256\repsep 289\repsep 346\repsep 347\repsep 348\repsep 369\repsep 579$ & 4480 \\
$015\repsep 024\repsep 028\repsep 035\repsep 038\repsep 124\repsep 126\repsep 147\repsep 289\repsep 349\repsep 357\repsep 679$ & 4160 \\
$017\repsep 034\repsep 058\repsep 146\repsep 257\repsep 259\repsep 269\repsep 356\repsep 379\repsep 468\repsep 489\repsep 678$ & 4160 \\
$017\repsep 034\repsep 038\repsep 057\repsep 058\repsep 123\repsep 146\repsep 159\repsep 245\repsep 278\repsep 389\repsep 567$ & 3968 \\
$026\repsep 049\repsep 058\repsep 068\repsep 124\repsep 137\repsep 168\repsep 235\repsep 348\repsep 357\repsep 369\repsep 567$ & 3968 \\
$012\repsep 035\repsep 039\repsep 058\repsep 067\repsep 146\repsep 235\repsep 245\repsep 278\repsep 347\repsep 369\repsep 489$ & 3648 \\
$019\repsep 037\repsep 048\repsep 056\repsep 124\repsep 129\repsep 135\repsep 238\repsep 279\repsep 349\repsep 467\repsep 589$ & 2886 \\
$013\repsep 025\repsep 036\repsep 049\repsep 068\repsep 145\repsep 189\repsep 239\repsep 248\repsep 267\repsep 346\repsep 579$ & 2710 \\
$056\repsep 123\repsep 157\repsep 168\repsep 239\repsep 245\repsep 267\repsep 289\repsep 346\repsep 357\repsep 358\repsep 379\repsep 478$ & 13920 \\
$023\repsep 049\repsep 056\repsep 126\repsep 139\repsep 145\repsep 156\repsep 247\repsep 346\repsep 357\repsep 379\repsep 589\repsep 679$ & 10752 \\
$018\repsep 024\repsep 035\repsep 068\repsep 079\repsep 137\repsep 145\repsep 158\repsep 257\repsep 346\repsep 389\repsep 478\repsep 589$ & 6240 \\
$024\repsep 038\repsep 059\repsep 137\repsep 145\repsep 168\repsep 249\repsep 289\repsep 349\repsep 459\repsep 467\repsep 578\repsep 678$ & 5056 \\
$018\repsep 025\repsep 047\repsep 069\repsep 137\repsep 156\repsep 238\repsep 249\repsep 267\repsep 269\repsep 345\repsep 467\repsep 468\smash{^\dagger}$ & 4224 \\
$027\repsep 036\repsep 049\repsep 078\repsep 129\repsep 135\repsep 234\repsep 249\repsep 268\repsep 289\repsep 378\repsep 458\repsep 569$ & 3968 \\
$019\repsep 028\repsep 036\repsep 048\repsep 057\repsep 138\repsep 147\repsep 156\repsep 178\repsep 279\repsep 469\repsep 589\repsep 678$ & 2976 \\
$019\repsep 034\repsep 047\repsep 058\repsep 067\repsep 127\repsep 138\repsep 236\repsep 249\repsep 356\repsep 379\repsep 478\repsep 689$ & 2818 \\
$012\repsep 035\repsep 067\repsep 129\repsep 145\repsep 179\repsep 246\repsep 257\repsep 289\repsep 368\repsep 379\repsep 478\repsep 589$ & 2451 \\
$019\repsep 027\repsep 034\repsep 079\repsep 128\repsep 137\repsep 146\repsep 235\repsep 267\repsep 389\repsep 457\repsep 568\repsep 579$ & 2451 \\
$016\repsep 028\repsep 057\repsep 123\repsep 178\repsep 249\repsep 268\repsep 345\repsep 367\repsep 468\repsep 479\repsep 589\repsep 679$ & 1874 \\
$028\repsep 037\repsep 078\repsep 136\repsep 139\repsep 168\repsep 234\repsep 239\repsep 257\repsep 346\repsep 356\repsep 469\repsep 478\repsep 589$ & 4864 \\
$012\repsep 037\repsep 048\repsep 059\repsep 134\repsep 156\repsep 159\repsep 179\repsep 236\repsep 249\repsep 367\repsep 389\repsep 467\repsep 679$ & 4576 \\
$026\repsep 028\repsep 036\repsep 129\repsep 135\repsep 159\repsep 234\repsep 267\repsep 356\repsep 378\repsep 456\repsep 489\repsep 579\repsep 678$ & 4224 \\
$013\repsep 018\repsep 023\repsep 026\repsep 035\repsep 078\repsep 079\repsep 124\repsep 169\repsep 237\repsep 349\repsep 368\repsep 478\repsep 567$ & 3648 \\
$026\repsep 034\repsep 057\repsep 089\repsep 123\repsep 125\repsep 147\repsep 178\repsep 259\repsep 356\repsep 379\repsep 458\repsep 568\repsep 678$ & 3520 \\
$035\repsep 037\repsep 047\repsep 069\repsep 128\repsep 145\repsep 237\repsep 246\repsep 259\repsep 349\repsep 359\repsep 478\repsep 567\repsep 789$ & 3428 \\
$015\repsep 024\repsep 036\repsep 078\repsep 123\repsep 127\repsep 129\repsep 189\repsep 238\repsep 256\repsep 349\repsep 458\repsep 467\repsep 478$ & 3248 \\
$017\repsep 036\repsep 049\repsep 145\repsep 146\repsep 234\repsep 268\repsep 279\repsep 357\repsep 389\repsep 489\repsep 569\repsep 578\repsep 579$ & 3236 \\
$019\repsep 024\repsep 036\repsep 037\repsep 078\repsep 148\repsep 157\repsep 238\repsep 279\repsep 345\repsep 358\repsep 467\repsep 468\repsep 689$ & 3108 \\
$013\repsep 014\repsep 029\repsep 035\repsep 068\repsep 123\repsep 248\repsep 257\repsep 278\repsep 345\repsep 349\repsep 378\repsep 456\repsep 679$ & 3072 \\
$017\repsep 034\repsep 046\repsep 067\repsep 068\repsep 137\repsep 148\repsep 159\repsep 239\repsep 246\repsep 278\repsep 358\repsep 457\repsep 679$ & 2588 \\
$018\repsep 026\repsep 047\repsep 059\repsep 127\repsep 136\repsep 149\repsep 259\repsep 267\repsep 278\repsep 357\repsep 359\repsep 389\repsep 456$ & 2512 \\
$018\repsep 036\repsep 079\repsep 124\repsep 138\repsep 157\repsep 169\repsep 235\repsep 246\repsep 248\repsep 289\repsep 346\repsep 478\repsep 568$ & 2460 \\
$012\repsep 039\repsep 048\repsep 078\repsep 136\repsep 156\repsep 157\repsep 246\repsep 279\repsep 347\repsep 358\repsep 378\repsep 569\repsep 678$ & 2404 \\
$013\repsep 016\repsep 034\repsep 058\repsep 059\repsep 069\repsep 128\repsep 135\repsep 237\repsep 246\repsep 369\repsep 458\repsep 567\repsep 789$ & 2252 \\
$017\repsep 036\repsep 048\repsep 057\repsep 089\repsep 123\repsep 156\repsep 247\repsep 258\repsep 269\repsep 345\repsep 379\repsep 679\repsep 689$ & 1879 \\
$027\repsep 036\repsep 047\repsep 048\repsep 058\repsep 124\repsep 135\repsep 169\repsep 239\repsep 378\repsep 456\repsep 467\repsep 489\repsep 579$ & 1753 \\
$013\repsep 024\repsep 058\repsep 067\repsep 129\repsep 156\repsep 159\repsep 235\repsep 246\repsep 268\repsep 347\repsep 389\repsep 469\repsep 579$ & 1737 \\
$012\repsep 036\repsep 047\repsep 089\repsep 157\repsep 159\repsep 167\repsep 238\repsep 246\repsep 259\repsep 349\repsep 458\repsep 459\repsep 578$ & 1688 \\
$039\repsep 048\repsep 057\repsep 123\repsep 158\repsep 167\repsep 239\repsep 245\repsep 279\repsep 349\repsep 368\repsep 389\repsep 456\repsep 469$ & 1408 \\
$014\repsep 018\repsep 023\repsep 058\repsep 069\repsep 079\repsep 135\repsep 189\repsep 247\repsep 256\repsep 379\repsep 458\repsep 459\repsep 678$ & 1368 \\
$016\repsep 035\repsep 079\repsep 124\repsep 159\repsep 189\repsep 237\repsep 246\repsep 259\repsep 268\repsep 279\repsep 348\repsep 469\repsep 567$ & 1356 \\
$013\repsep 028\repsep 069\repsep 127\repsep 146\repsep 158\repsep 239\repsep 245\repsep 357\repsep 359\repsep 458\repsep 467\repsep 479\repsep 678$ & 1258 \\
$014\repsep 038\repsep 067\repsep 124\repsep 136\repsep 138\repsep 178\repsep 234\repsep 257\repsep 289\repsep 359\repsep 456\repsep 569\repsep 679$ & 1250 \\
$015\repsep 017\repsep 034\repsep 089\repsep 126\repsep 135\repsep 148\repsep 237\repsep 249\repsep 256\repsep 258\repsep 457\repsep 569\repsep 678$ & 1240 \\
$017\repsep 029\repsep 034\repsep 037\repsep 124\repsep 138\repsep 156\repsep 246\repsep 258\repsep 346\repsep 379\repsep 457\repsep 458\repsep 689$ & 1236 \\
$017\repsep 019\repsep 026\repsep 058\repsep 138\repsep 156\repsep 247\repsep 257\repsep 289\repsep 345\repsep 346\repsep 367\repsep 469\repsep 579$ & 1223 \\
$017\repsep 023\repsep 034\repsep 058\repsep 069\repsep 128\repsep 136\repsep 178\repsep 245\repsep 279\repsep 359\repsep 478\repsep 489\repsep 567$ & 1217 \\
$015\repsep 034\repsep 039\repsep 067\repsep 126\repsep 139\repsep 159\repsep 178\repsep 247\repsep 248\repsep 259\repsep 358\repsep 457\repsep 689$ & 1216 \\
$015\repsep 018\repsep 027\repsep 034\repsep 069\repsep 126\repsep 147\repsep 234\repsep 235\repsep 368\repsep 456\repsep 458\repsep 489\repsep 579$ & 1197 \\
$016\repsep 023\repsep 036\repsep 049\repsep 057\repsep 127\repsep 134\repsep 138\repsep 159\repsep 258\repsep 356\repsep 379\repsep 467\repsep 689$ & 1190 \\
$023\repsep 039\repsep 045\repsep 089\repsep 147\repsep 156\repsep 189\repsep 234\repsep 238\repsep 258\repsep 267\repsep 357\repsep 468\repsep 479$ & 820 \\
$029\repsep 038\repsep 045\repsep 123\repsep 148\repsep 159\repsep 189\repsep 246\repsep 278\repsep 356\repsep 359\repsep 378\repsep 379\repsep 456\repsep 689$ & 3072 \\
$019\repsep 025\repsep 034\repsep 057\repsep 123\repsep 128\repsep 136\repsep 145\repsep 237\repsep 389\repsep 478\repsep 489\repsep 567\repsep 569\repsep 689$ & 2448 \\
$025\repsep 028\repsep 047\repsep 128\repsep 129\repsep 136\repsep 148\repsep 156\repsep 179\repsep 236\repsep 359\repsep 378\repsep 469\repsep 568\repsep 678$ & 2281 \\
$016\repsep 028\repsep 034\repsep 037\repsep 079\repsep 135\repsep 189\repsep 237\repsep 249\repsep 258\repsep 259\repsep 459\repsep 468\repsep 567\repsep 569$ & 2276 \\
$017\repsep 028\repsep 029\repsep 069\repsep 123\repsep 137\repsep 146\repsep 248\repsep 256\repsep 289\repsep 347\repsep 378\repsep 389\repsep 459\repsep 578$ & 2256 \\
$012\repsep 017\repsep 038\repsep 057\repsep 069\repsep 134\repsep 167\repsep 237\repsep 245\repsep 268\repsep 278\repsep 356\repsep 456\repsep 478\repsep 589$ & 2087 \\
$025\repsep 039\repsep 046\repsep 078\repsep 124\repsep 127\repsep 169\repsep 234\repsep 357\repsep 368\repsep 458\repsep 467\repsep 479\repsep 579\repsep 589$ & 2062 \\
$014\repsep 026\repsep 039\repsep 078\repsep 079\repsep 137\repsep 158\repsep 167\repsep 169\repsep 234\repsep 268\repsep 279\repsep 459\repsep 467\repsep 568$ & 2044 \\
$018\repsep 027\repsep 036\repsep 045\repsep 125\repsep 129\repsep 134\repsep 157\repsep 235\repsep 349\repsep 378\repsep 468\repsep 479\repsep 589\repsep 789$ & 2044 \\
$012\repsep 016\repsep 027\repsep 056\repsep 089\repsep 129\repsep 135\repsep 147\repsep 234\repsep 345\repsep 367\repsep 389\repsep 459\repsep 478\repsep 568$ & 2026 \\
$012\repsep 029\repsep 037\repsep 049\repsep 068\repsep 124\repsep 147\repsep 159\repsep 239\repsep 256\repsep 278\repsep 346\repsep 358\repsep 367\repsep 789$ & 2026 \\
$016\repsep 023\repsep 049\repsep 057\repsep 069\repsep 127\repsep 148\repsep 238\repsep 246\repsep 289\repsep 347\repsep 357\repsep 359\repsep 368\repsep 679$ & 2026 \\
$012\repsep 036\repsep 058\repsep 079\repsep 139\repsep 157\repsep 168\repsep 179\repsep 237\repsep 245\repsep 289\repsep 369\repsep 457\repsep 469\repsep 478$ & 2019 \\
$017\repsep 039\repsep 058\repsep 059\repsep 126\repsep 135\repsep 138\repsep 145\repsep 234\repsep 235\repsep 259\repsep 278\repsep 467\repsep 489\repsep 579$ & 2017 \\
$025\repsep 049\repsep 127\repsep 135\repsep 146\repsep 189\repsep 234\repsep 269\repsep 346\repsep 357\repsep 368\repsep 379\repsep 478\repsep 567\repsep 678$ & 2002 \\
$017\repsep 024\repsep 038\repsep 059\repsep 067\repsep 139\repsep 148\repsep 156\repsep 235\repsep 237\repsep 247\repsep 257\repsep 379\repsep 467\repsep 689$ & 1998 \\
$013\repsep 027\repsep 034\repsep 049\repsep 058\repsep 128\repsep 147\repsep 156\repsep 237\repsep 239\repsep 289\repsep 357\repsep 459\repsep 468\repsep 789$ & 1996 \\
$017\repsep 037\repsep 039\repsep 058\repsep 069\repsep 124\repsep 138\repsep 169\repsep 236\repsep 345\repsep 457\repsep 458\repsep 489\repsep 579\repsep 678$ & 1996 \\
$015\repsep 027\repsep 046\repsep 068\repsep 089\repsep 124\repsep 169\repsep 178\repsep 249\repsep 268\repsep 347\repsep 359\repsep 367\repsep 456\repsep 458$ & 1991 \\
$014\repsep 037\repsep 046\repsep 059\repsep 068\repsep 127\repsep 135\repsep 146\repsep 169\repsep 179\repsep 239\repsep 245\repsep 368\repsep 479\repsep 578$ & 1983 \\
$034\repsep 035\repsep 056\repsep 078\repsep 124\repsep 135\repsep 169\repsep 234\repsep 236\repsep 257\repsep 289\repsep 379\repsep 457\repsep 459\repsep 467$ & 1969 \\
$014\repsep 017\repsep 034\repsep 059\repsep 129\repsep 145\repsep 156\repsep 168\repsep 234\repsep 238\repsep 256\repsep 379\repsep 478\repsep 678\repsep 679$ & 1897 \\
$017\repsep 024\repsep 036\repsep 037\repsep 079\repsep 089\repsep 123\repsep 168\repsep 179\repsep 236\repsep 257\repsep 378\repsep 458\repsep 467\repsep 569$ & 1874 \\
$035\repsep 048\repsep 067\repsep 129\repsep 136\repsep 139\repsep 157\repsep 237\repsep 245\repsep 268\repsep 389\repsep 479\repsep 569\repsep 578\repsep 689$ & 1813 \\
$018\repsep 045\repsep 057\repsep 067\repsep 068\repsep 127\repsep 135\repsep 146\repsep 236\repsep 289\repsep 349\repsep 358\repsep 456\repsep 458\repsep 679$ & 1759 \\
$012\repsep 026\repsep 038\repsep 057\repsep 127\repsep 134\repsep 139\repsep 169\repsep 245\repsep 258\repsep 269\repsep 356\repsep 468\repsep 479\repsep 589$ & 1729 \\
$027\repsep 035\repsep 039\repsep 045\repsep 069\repsep 125\repsep 148\repsep 149\repsep 168\repsep 234\repsep 367\repsep 389\repsep 458\repsep 579\repsep 678$ & 1702 \\
$018\repsep 023\repsep 024\repsep 037\repsep 058\repsep 067\repsep 127\repsep 136\repsep 145\repsep 179\repsep 249\repsep 357\repsep 389\repsep 468\repsep 569$ & 1690 \\
$017\repsep 027\repsep 034\repsep 056\repsep 057\repsep 123\repsep 139\repsep 158\repsep 169\repsep 178\repsep 246\repsep 259\repsep 357\repsep 368\repsep 489$ & 1684 \\
$017\repsep 018\repsep 035\repsep 089\repsep 125\repsep 126\repsep 148\repsep 169\repsep 234\repsep 268\repsep 279\repsep 367\repsep 456\repsep 459\repsep 578$ & 1683 \\
$012\repsep 018\repsep 029\repsep 057\repsep 125\repsep 137\repsep 157\repsep 169\repsep 236\repsep 348\repsep 359\repsep 389\repsep 456\repsep 479\repsep 678$ & 1674 \\
$018\repsep 025\repsep 026\repsep 035\repsep 037\repsep 078\repsep 089\repsep 139\repsep 146\repsep 157\repsep 234\repsep 279\repsep 368\repsep 458\repsep 569$ & 1673 \\
$016\repsep 027\repsep 029\repsep 038\repsep 045\repsep 127\repsep 134\repsep 169\repsep 235\repsep 268\repsep 369\repsep 457\repsep 479\repsep 567\repsep 589$ & 1668 \\
$016\repsep 029\repsep 037\repsep 048\repsep 127\repsep 129\repsep 149\repsep 235\repsep 346\repsep 389\repsep 457\repsep 478\repsep 568\repsep 579\repsep 589$ & 1663 \\
$018\repsep 025\repsep 028\repsep 037\repsep 046\repsep 136\repsep 149\repsep 159\repsep 239\repsep 245\repsep 267\repsep 358\repsep 478\repsep 569\repsep 579$ & 1656 \\
$016\repsep 024\repsep 035\repsep 079\repsep 127\repsep 147\repsep 156\repsep 159\repsep 238\repsep 239\repsep 346\repsep 457\repsep 489\repsep 568\repsep 789$ & 1654 \\
$012\repsep 035\repsep 047\repsep 058\repsep 078\repsep 138\repsep 139\repsep 157\repsep 236\repsep 247\repsep 289\repsep 368\repsep 459\repsep 468\repsep 679$ & 1646 \\
\bottomrule
\end{tabular}
\end{minipage}\hfill
\begin{minipage}[t]{0.495\textwidth}
\begin{tabular}{lr}
\toprule
representative & $K/2^{10}$ \\
\midrule
$023\repsep 034\repsep 059\repsep 078\repsep 125\repsep 137\repsep 149\repsep 168\repsep 234\repsep 236\repsep 248\repsep 258\repsep 267\repsep 389\repsep 456$ & 1646 \\
$013\repsep 025\repsep 038\repsep 046\repsep 049\repsep 057\repsep 146\repsep 159\repsep 178\repsep 235\repsep 268\repsep 279\repsep 348\repsep 367\repsep 489$ & 1636 \\
$018\repsep 025\repsep 039\repsep 046\repsep 135\repsep 149\repsep 234\repsep 237\repsep 268\repsep 367\repsep 478\repsep 567\repsep 578\repsep 589\repsep 679$ & 1633 \\
$012\repsep 013\repsep 014\repsep 027\repsep 035\repsep 038\repsep 059\repsep 129\repsep 136\repsep 148\repsep 234\repsep 379\repsep 456\repsep 578\repsep 689$ & 1631 \\
$012\repsep 049\repsep 058\repsep 059\repsep 068\repsep 137\repsep 139\repsep 167\repsep 234\repsep 256\repsep 279\repsep 357\repsep 458\repsep 469\repsep 789$ & 1620 \\
$025\repsep 038\repsep 067\repsep 079\repsep 147\repsep 158\repsep 169\repsep 237\repsep 246\repsep 259\repsep 267\repsep 345\repsep 348\repsep 349\repsep 789$ & 1340 \\
$025\repsep 037\repsep 068\repsep 078\repsep 128\repsep 146\repsep 149\repsep 156\repsep 178\repsep 247\repsep 348\repsep 367\repsep 369\repsep 579\repsep 589$ & 1332 \\
$019\repsep 025\repsep 027\repsep 034\repsep 068\repsep 135\repsep 137\repsep 156\repsep 237\repsep 248\repsep 269\repsep 278\repsep 456\repsep 479\repsep 578$ & 1305 \\
$029\repsep 036\repsep 037\repsep 038\repsep 058\repsep 128\repsep 145\repsep 179\repsep 238\repsep 246\repsep 269\repsep 357\repsep 489\repsep 569\repsep 678$ & 1303 \\
$018\repsep 037\repsep 057\repsep 069\repsep 127\repsep 135\repsep 149\repsep 239\repsep 246\repsep 247\repsep 269\repsep 458\repsep 469\repsep 568\repsep 678$ & 1299 \\
$023\repsep 046\repsep 058\repsep 079\repsep 126\repsep 138\repsep 139\repsep 157\repsep 248\repsep 259\repsep 345\repsep 367\repsep 589\repsep 678\repsep 789$ & 1291 \\
$026\repsep 036\repsep 049\repsep 057\repsep 128\repsep 137\repsep 169\repsep 189\repsep 235\repsep 245\repsep 248\repsep 346\repsep 458\repsep 567\repsep 789$ & 1273 \\
$013\repsep 023\repsep 026\repsep 059\repsep 128\repsep 145\repsep 146\repsep 158\repsep 279\repsep 347\repsep 349\repsep 367\repsep 456\repsep 478\repsep 689$ & 1271 \\
$014\repsep 025\repsep 039\repsep 078\repsep 125\repsep 157\repsep 168\repsep 179\repsep 237\repsep 269\repsep 348\repsep 356\repsep 368\repsep 457\repsep 459$ & 1271 \\
$029\repsep 037\repsep 058\repsep 127\repsep 135\repsep 149\repsep 156\repsep 246\repsep 368\repsep 389\repsep 457\repsep 469\repsep 568\repsep 678\repsep 679$ & 1268 \\
$026\repsep 029\repsep 046\repsep 048\repsep 079\repsep 127\repsep 135\repsep 169\repsep 258\repsep 289\repsep 346\repsep 389\repsep 456\repsep 457\repsep 678$ & 1261 \\
$034\repsep 057\repsep 058\repsep 067\repsep 128\repsep 134\repsep 169\repsep 235\repsep 249\repsep 258\repsep 357\repsep 368\repsep 379\repsep 457\repsep 678$ & 1261 \\
$014\repsep 037\repsep 038\repsep 048\repsep 056\repsep 079\repsep 123\repsep 149\repsep 157\repsep 248\repsep 269\repsep 347\repsep 358\repsep 359\repsep 678$ & 1254 \\
$015\repsep 026\repsep 029\repsep 037\repsep 045\repsep 127\repsep 145\repsep 157\repsep 189\repsep 235\repsep 238\repsep 346\repsep 478\repsep 568\repsep 679$ & 1254 \\
$016\repsep 027\repsep 046\repsep 049\repsep 058\repsep 089\repsep 136\repsep 137\repsep 149\repsep 168\repsep 239\repsep 245\repsep 348\repsep 568\repsep 579$ & 1246 \\
$015\repsep 029\repsep 035\repsep 068\repsep 128\repsep 134\repsep 167\repsep 236\repsep 248\repsep 258\repsep 348\repsep 357\repsep 468\repsep 469\repsep 789$ & 1244 \\
$014\repsep 023\repsep 036\repsep 089\repsep 129\repsep 156\repsep 178\repsep 245\repsep 246\repsep 258\repsep 279\repsep 348\repsep 357\repsep 369\repsep 459$ & 1239 \\
$016\repsep 027\repsep 034\repsep 048\repsep 089\repsep 128\repsep 135\repsep 149\repsep 179\repsep 236\repsep 368\repsep 379\repsep 458\repsep 478\repsep 569$ & 1238 \\
$017\repsep 025\repsep 034\repsep 048\repsep 069\repsep 123\repsep 126\repsep 129\repsep 189\repsep 247\repsep 358\repsep 367\repsep 378\repsep 459\repsep 468$ & 1238 \\
$013\repsep 029\repsep 045\repsep 068\repsep 123\repsep 124\repsep 156\repsep 178\repsep 238\repsep 257\repsep 345\repsep 359\repsep 367\repsep 378\repsep 469$ & 1237 \\
$019\repsep 023\repsep 057\repsep 068\repsep 125\repsep 145\repsep 146\repsep 249\repsep 267\repsep 348\repsep 368\repsep 379\repsep 469\repsep 478\repsep 589$ & 1232 \\
$015\repsep 016\repsep 023\repsep 048\repsep 079\repsep 129\repsep 136\repsep 237\repsep 245\repsep 349\repsep 368\repsep 458\repsep 467\repsep 569\repsep 578$ & 1230 \\
$015\repsep 024\repsep 036\repsep 059\repsep 126\repsep 134\repsep 135\repsep 136\repsep 178\repsep 237\repsep 258\repsep 346\repsep 379\repsep 457\repsep 689$ & 1228 \\
$016\repsep 019\repsep 056\repsep 078\repsep 079\repsep 145\repsep 147\repsep 168\repsep 239\repsep 245\repsep 267\repsep 348\repsep 357\repsep 378\repsep 569$ & 1228 \\
$018\repsep 026\repsep 039\repsep 057\repsep 125\repsep 128\repsep 145\repsep 147\repsep 169\repsep 238\repsep 249\repsep 357\repsep 367\repsep 468\repsep 589$ & 1227 \\
$017\repsep 018\repsep 035\repsep 038\repsep 067\repsep 127\repsep 149\repsep 156\repsep 234\repsep 289\repsep 349\repsep 368\repsep 379\repsep 456\repsep 458$ & 1226 \\
$012\repsep 034\repsep 035\repsep 047\repsep 056\repsep 135\repsep 146\repsep 149\repsep 168\repsep 179\repsep 239\repsep 256\repsep 278\repsep 458\repsep 689$ & 1224 \\
$015\repsep 023\repsep 029\repsep 036\repsep 089\repsep 128\repsep 147\repsep 235\repsep 246\repsep 279\repsep 349\repsep 378\repsep 458\repsep 567\repsep 579$ & 1221 \\
$016\repsep 024\repsep 036\repsep 057\repsep 089\repsep 129\repsep 136\repsep 138\repsep 147\repsep 158\repsep 179\repsep 236\repsep 278\repsep 345\repsep 469$ & 1215 \\
$024\repsep 025\repsep 046\repsep 079\repsep 125\repsep 127\repsep 138\repsep 159\repsep 178\repsep 248\repsep 345\repsep 348\repsep 359\repsep 369\repsep 678$ & 1206 \\
$012\repsep 045\repsep 046\repsep 078\repsep 079\repsep 125\repsep 149\repsep 168\repsep 234\repsep 256\repsep 278\repsep 357\repsep 369\repsep 589\repsep 789$ & 1205 \\
$013\repsep 058\repsep 059\repsep 079\repsep 124\repsep 126\repsep 148\repsep 156\repsep 234\repsep 289\repsep 357\repsep 378\repsep 467\repsep 567\repsep 569$ & 1199 \\
$017\repsep 034\repsep 089\repsep 125\repsep 139\repsep 149\repsep 236\repsep 248\repsep 279\repsep 357\repsep 359\repsep 456\repsep 467\repsep 479\repsep 568$ & 1195 \\
$026\repsep 034\repsep 057\repsep 089\repsep 123\repsep 145\repsep 146\repsep 148\repsep 158\repsep 179\repsep 258\repsep 369\repsep 378\repsep 467\repsep 689$ & 1191 \\
$014\repsep 036\repsep 057\repsep 089\repsep 124\repsep 139\repsep 169\repsep 178\repsep 238\repsep 256\repsep 345\repsep 356\repsep 379\repsep 467\repsep 479$ & 1165 \\
$019\repsep 046\repsep 048\repsep 057\repsep 124\repsep 145\repsep 156\repsep 168\repsep 235\repsep 267\repsep 349\repsep 378\repsep 389\repsep 568\repsep 589$ & 778 \\
$013\repsep 024\repsep 027\repsep 036\repsep 057\repsep 124\repsep 128\repsep 149\repsep 267\repsep 359\repsep 378\repsep 456\repsep 459\repsep 678\repsep 689$ & 771 \\
$015\repsep 018\repsep 029\repsep 059\repsep 069\repsep 078\repsep 127\repsep 138\repsep 149\repsep 268\repsep 346\repsep 367\repsep 379\repsep 458\repsep 567$ & 767 \\
$017\repsep 025\repsep 036\repsep 058\repsep 078\repsep 089\repsep 125\repsep 138\repsep 149\repsep 245\repsep 247\repsep 248\repsep 267\repsep 359\repsep 567\repsep 789$ & 2354 \\
$013\repsep 023\repsep 024\repsep 059\repsep 078\repsep 128\repsep 135\repsep 146\repsep 157\repsep 236\repsep 279\repsep 347\repsep 368\repsep 378\repsep 457\repsep 458$ & 2322 \\
$018\repsep 024\repsep 035\repsep 036\repsep 057\repsep 126\repsep 127\repsep 135\repsep 189\repsep 238\repsep 267\repsep 268\repsep 289\repsep 379\repsep 469\repsep 478$ & 2063 \\
$019\repsep 026\repsep 046\repsep 058\repsep 124\repsep 125\repsep 147\repsep 148\repsep 168\repsep 259\repsep 267\repsep 278\repsep 348\repsep 349\repsep 357\repsep 789$ & 2027 \\
$018\repsep 049\repsep 057\repsep 089\repsep 135\repsep 146\repsep 169\repsep 237\repsep 246\repsep 247\repsep 259\repsep 289\repsep 346\repsep 368\repsep 478\repsep 568$ & 1868 \\
$017\repsep 029\repsep 035\repsep 046\repsep 139\repsep 156\repsep 167\repsep 189\repsep 234\repsep 258\repsep 267\repsep 368\repsep 379\repsep 459\repsep 469\repsep 578$ & 1795 \\
$024\repsep 039\repsep 057\repsep 068\repsep 128\repsep 135\repsep 149\repsep 178\repsep 238\repsep 239\repsep 247\repsep 267\repsep 367\repsep 458\repsep 569\repsep 579$ & 1721 \\
$023\repsep 048\repsep 059\repsep 067\repsep 123\repsep 129\repsep 145\repsep 168\repsep 178\repsep 237\repsep 258\repsep 346\repsep 349\repsep 356\repsep 456\repsep 789$ & 1708 \\
$029\repsep 034\repsep 056\repsep 078\repsep 079\repsep 134\repsep 138\repsep 146\repsep 179\repsep 237\repsep 245\repsep 268\repsep 359\repsep 456\repsep 459\repsep 489$ & 1695 \\
$015\repsep 037\repsep 056\repsep 068\repsep 124\repsep 169\repsep 178\repsep 189\repsep 235\repsep 279\repsep 348\repsep 369\repsep 379\repsep 456\repsep 459\repsep 578$ & 1690 \\
$023\repsep 029\repsep 036\repsep 047\repsep 058\repsep 089\repsep 128\repsep 134\repsep 156\repsep 179\repsep 249\repsep 268\repsep 358\repsep 367\repsep 457\repsep 678$ & 1685 \\
$017\repsep 019\repsep 034\repsep 046\repsep 059\repsep 068\repsep 124\repsep 127\repsep 158\repsep 238\repsep 267\repsep 289\repsep 357\repsep 369\repsep 456\repsep 478$ & 1683 \\
$023\repsep 027\repsep 036\repsep 058\repsep 079\repsep 126\repsep 138\repsep 147\repsep 148\repsep 149\repsep 235\repsep 246\repsep 249\repsep 457\repsep 569\repsep 789$ & 1683 \\
$012\repsep 037\repsep 046\repsep 059\repsep 124\repsep 145\repsep 169\repsep 178\repsep 239\repsep 248\repsep 249\repsep 267\repsep 356\repsep 378\repsep 389\repsep 479$ & 1681 \\
$016\repsep 028\repsep 029\repsep 035\repsep 037\repsep 048\repsep 123\repsep 124\repsep 139\repsep 159\repsep 258\repsep 269\repsep 346\repsep 378\repsep 479\repsep 567$ & 1680 \\
$019\repsep 047\repsep 056\repsep 089\repsep 126\repsep 128\repsep 134\repsep 159\repsep 235\repsep 236\repsep 278\repsep 357\repsep 369\repsep 379\repsep 458\repsep 459$ & 1680 \\
$014\repsep 023\repsep 029\repsep 035\repsep 039\repsep 067\repsep 158\repsep 179\repsep 234\repsep 238\repsep 246\repsep 369\repsep 378\repsep 459\repsep 478\repsep 567$ & 1676 \\
$012\repsep 038\repsep 048\repsep 059\repsep 067\repsep 068\repsep 139\repsep 146\repsep 167\repsep 234\repsep 258\repsep 267\repsep 357\repsep 378\repsep 459\repsep 479$ & 1670 \\
$014\repsep 046\repsep 056\repsep 078\repsep 123\repsep 127\repsep 159\repsep 168\repsep 235\repsep 249\repsep 347\repsep 368\repsep 389\repsep 467\repsep 468\repsep 679$ & 1668 \\
$023\repsep 028\repsep 035\repsep 047\repsep 069\repsep 125\repsep 139\repsep 158\repsep 167\repsep 237\repsep 249\repsep 269\repsep 346\repsep 379\repsep 458\repsep 789$ & 1667 \\
$015\repsep 019\repsep 025\repsep 036\repsep 124\repsep 146\repsep 168\repsep 189\repsep 237\repsep 249\repsep 279\repsep 345\repsep 389\repsep 478\repsep 568\repsep 679$ & 1650 \\
$023\repsep 047\repsep 048\repsep 089\repsep 126\repsep 127\repsep 128\repsep 135\repsep 148\repsep 257\repsep 259\repsep 289\repsep 368\repsep 379\repsep 469\repsep 567$ & 1648 \\
$012\repsep 014\repsep 034\repsep 036\repsep 057\repsep 089\repsep 125\repsep 146\repsep 157\repsep 178\repsep 237\repsep 248\repsep 269\repsep 359\repsep 456\repsep 479$ & 1643 \\
$019\repsep 023\repsep 034\repsep 035\repsep 048\repsep 067\repsep 136\repsep 178\repsep 245\repsep 258\repsep 268\repsep 278\repsep 279\repsep 347\repsep 469\repsep 589$ & 1642 \\
$013\repsep 026\repsep 029\repsep 057\repsep 089\repsep 123\repsep 158\repsep 167\repsep 169\repsep 189\repsep 239\repsep 245\repsep 278\repsep 346\repsep 479\repsep 678$ & 1641 \\
$018\repsep 029\repsep 034\repsep 056\repsep 069\repsep 124\repsep 126\repsep 137\repsep 169\repsep 258\repsep 267\repsep 359\repsep 468\repsep 489\repsep 579\repsep 789$ & 1629 \\
$014\repsep 019\repsep 026\repsep 039\repsep 046\repsep 058\repsep 129\repsep 137\repsep 138\repsep 159\repsep 237\repsep 245\repsep 289\repsep 469\repsep 478\repsep 567$ & 1627 \\
$018\repsep 034\repsep 056\repsep 068\repsep 079\repsep 127\repsep 146\repsep 159\repsep 169\repsep 179\repsep 236\repsep 248\repsep 257\repsep 389\repsep 459\repsep 578$ & 1616 \\
$012\repsep 035\repsep 037\repsep 069\repsep 136\repsep 157\repsep 169\repsep 189\repsep 245\repsep 279\repsep 349\repsep 357\repsep 368\repsep 379\repsep 478\repsep 568$ & 1613 \\
$019\repsep 027\repsep 036\repsep 048\repsep 126\repsep 137\repsep 146\repsep 158\repsep 234\repsep 256\repsep 359\repsep 368\repsep 379\repsep 389\repsep 457\repsep 789$ & 1609 \\
$016\repsep 048\repsep 057\repsep 067\repsep 123\repsep 128\repsep 178\repsep 245\repsep 249\repsep 268\repsep 269\repsep 345\repsep 356\repsep 389\repsep 459\repsep 479$ & 1279 \\
$023\repsep 037\repsep 046\repsep 057\repsep 089\repsep 129\repsep 134\repsep 168\repsep 189\repsep 236\repsep 257\repsep 268\repsep 357\repsep 358\repsep 478\repsep 569$ & 1255 \\
$014\repsep 058\repsep 079\repsep 125\repsep 126\repsep 128\repsep 135\repsep 146\repsep 178\repsep 237\repsep 249\repsep 345\repsep 369\repsep 378\repsep 468\repsep 567$ & 1246 \\
$013\repsep 048\repsep 058\repsep 079\repsep 136\repsep 149\repsep 157\repsep 167\repsep 234\repsep 237\repsep 256\repsep 278\repsep 359\repsep 367\repsep 579\repsep 689$ & 1240 \\
$016\repsep 035\repsep 037\repsep 048\repsep 079\repsep 124\repsep 138\repsep 157\repsep 169\repsep 246\repsep 256\repsep 278\repsep 279\repsep 289\repsep 349\repsep 467$ & 1236 \\
$012\repsep 018\repsep 057\repsep 069\repsep 079\repsep 156\repsep 159\repsep 189\repsep 234\repsep 258\repsep 356\repsep 379\repsep 467\repsep 489\repsep 589\repsep 678$ & 1216 \\
$013\repsep 019\repsep 025\repsep 048\repsep 135\repsep 147\repsep 168\repsep 236\repsep 239\repsep 245\repsep 289\repsep 346\repsep 358\repsep 379\repsep 567\repsep 679$ & 1208 \\
$023\repsep 048\repsep 049\repsep 056\repsep 079\repsep 128\repsep 136\repsep 137\repsep 159\repsep 179\repsep 246\repsep 249\repsep 257\repsep 358\repsep 467\repsep 678$ & 1195 \\
$017\repsep 037\repsep 038\repsep 046\repsep 056\repsep 124\repsep 139\repsep 168\repsep 246\repsep 258\repsep 269\repsep 345\repsep 367\repsep 378\repsep 489\repsep 579$ & 1189 \\
$013\repsep 018\repsep 029\repsep 067\repsep 123\repsep 138\repsep 156\repsep 179\repsep 235\repsep 245\repsep 247\repsep 256\repsep 345\repsep 389\repsep 469\repsep 578$ & 1183 \\
$015\repsep 019\repsep 026\repsep 039\repsep 048\repsep 127\repsep 169\repsep 235\repsep 239\repsep 249\repsep 347\repsep 368\repsep 379\repsep 456\repsep 457\repsep 589$ & 861 \\
$029\repsep 035\repsep 047\repsep 068\repsep 125\repsep 127\repsep 146\repsep 148\repsep 179\repsep 267\repsep 348\repsep 357\repsep 369\repsep 378\repsep 457\repsep 459$ & 833 \\
$025\repsep 049\repsep 078\repsep 089\repsep 127\repsep 138\repsep 145\repsep 159\repsep 169\repsep 234\repsep 359\repsep 367\repsep 379\repsep 467\repsep 468\repsep 589$ & 762 \\
$035\repsep 046\repsep 089\repsep 124\repsep 135\repsep 136\repsep 147\repsep 148\repsep 158\repsep 179\repsep 237\repsep 256\repsep 259\repsep 348\repsep 457\repsep 678$ & 743 \\
$015\repsep 024\repsep 079\repsep 128\repsep 137\repsep 169\repsep 179\repsep 235\repsep 236\repsep 267\repsep 348\repsep 456\repsep 468\repsep 589\repsep 678\repsep 679$ & 451 \\
$015\repsep 025\repsep 026\repsep 034\repsep 078\repsep 128\repsep 147\repsep 236\repsep 245\repsep 279\repsep 357\repsep 369\repsep 456\repsep 468\repsep 579\repsep 589$ & 413 \\
\phantom{}\\
\bottomrule
\end{tabular}
\end{minipage}
\end{table*}